\documentclass[12pt]{article}
\usepackage{amsmath,amssymb,amsthm}
\usepackage{pstricks,pst-node}
\usepackage{algorithm}
\usepackage{algorithmic}

\newtheorem{theorem}{Theorem}[section]
\newtheorem{lemma}[theorem]{Lemma}

\newcommand\co{\mathrm{co}}
\newcommand\sco{\mathcal{CO}}
\newcommand\cc{\mathrm{cc}}
\newcommand\scc{\mathcal{CC}}

\begin{document}
\date{\today}
\title{Convex sets in acyclic digraphs}
\author{Paul Balister\thanks{Department of Mathematical Sciences,
University of Memphis, TN 38152-3240, USA. E-mail:
pbalistr@memphis.edu} $^\S$ \and Stefanie Gerke\thanks{Department of
Mathematics, Royal Holloway, University of London, Egham, TW20 0EX,
UK, E-mail: stefanie.gerke@rhul.ac.uk} \and Gregory
Gutin\thanks{Department of Computer Science, Royal Holloway,
University of London, Egham, TW20 0EX, UK, E-mail:
gregory.gutin@rhul.ac.uk} } \maketitle
\begin{abstract}
A non-empty set $X$ of vertices of an acyclic digraph is called
connected if the underlying undirected graph induced by $X$ is
connected and it is called convex if no two vertices of $X$ are
connected by a directed path in which some vertices are not in
$X$. The set of convex sets (connected convex sets) of an acyclic
digraph $D$ is denoted by $\sco(D)$ ($\scc(D)$) and its size by
$\co(D)$ ($\cc(D)$). Gutin, Johnstone, Reddington, Scott,
Soleimanfallah, and Yeo  (Proc. ACiD'07)  conjectured that the sum
of the sizes of all (connected) convex sets in $D$ equals
$\Theta(n \cdot \co(D))$ ($\Theta(n \cdot \cc(D))$) where $n$ is
the order of $D$.

In this paper we exhibit a family of connected acyclic digraphs with
$\sum_{C\in \sco(D)}|C| = o(n\cdot \co(D))$ and $\sum_{C\in
\scc(D)}|C| = o(n\cdot \cc(D))$.  We also show that the number of
connected convex sets of order $k$ in any connected acyclic digraph
of order $n$ is at least $n-k+1$. This is a strengthening of a
theorem by Gutin and Yeo.
\end{abstract}

\section{Introduction}
Let $D$ be an acyclic digraph of order $n$. A non-empty set $X$ of
vertices in $D$ is {\em connected} if the underlying undirected
graph of $D[X]$, the subgraph of $D$ induced by $X$, is connected.
A non-empty set $X$ of vertices in $D$ is {\em convex} if there is
no directed path in $D$ between vertices of $X$ containing a
vertex not in $X$. The set of all convex sets of $D$ is denoted by
$\sco(D)$ and its size by $\co(D)$. The set of all connected
convex sets of $D$ is denoted by $\scc(D)$ and its size by
$\cc(D)$. Convex sets and connected convex sets in acyclic
digraphs are of interest in the field of custom computing in which
central processor architectures are parameterized for particular
applications, see, e.g., \cite{chen,gjrssy07}.

Gutin, Johnstone, Reddington, Scott, Soleimanfallah, and Yeo
\cite{gjrssy07} introduced an algorithm $\cal A$ determining all
connected convex sets of $D$ in time $O(n\cdot \cc(D)).$ They
observed that $\cal A$ can be modified to produce all convex sets in
time $O(n\cdot \co(D)).$ The authors of \cite{gjrssy07} conjectured
that the sum of the sizes of all convex sets (all connected convex
sets, respectively) in $D$ equals $\Theta(n \cdot \co(D))$
($\Theta(n\cdot \cc(D))$, respectively). If the conjecture were
true, then their algorithms would be optimal. The conjecture can be
formulated differently. Let $\bar{s}_{\co}(D)$ and
$\bar{s}_{\cc}(D)$ be the average size of a convex set and the
average size of a connected convex set in $D$. The conjecture claims
that $\bar{s}_{\co}(D)=\Theta(n)$ and $\bar{s}_{\cc}(D)=\Theta(n).$

In this paper we disprove both parts of the conjecture by exhibiting
a family $\cal F$ of digraphs for which
$\bar{s}_{\co}(D)=O(\sqrt{n})$ and $\bar{s}_{\cc}(D)=O(\sqrt{n})$;
see Section~\ref{sec:example}. In Section~\ref{sec:size} we show
that each connected digraph of order $n$ contains at least  $n-k+1$
connected convex sets of size $k$ for each $1\leq k \leq n$. This
extends a result of Gutin and Yeo \cite{gy07} who showed that each
connected acyclic digraph of order $n$ has at least $n(n+1)/2$
connected convex sets.

To simplify notation in the rest of the paper, we use $n$ for the
order of the digraph under consideration; $[m]$ will denote the
set $\{1,2,\ldots ,m\}$ ($m$ is a positive integer). A vertex $x$
of $D$ is a {\em source} ({\em sink}) if its in-degree $d^-(x)$
(out-degree $d^+(x)$) equals zero. A vertex $v$ of a connected
digraph $D$ is a {\em cut-vertex} if $D-v$ is not connected, i.e.,
$V(D)-v$ is not connected.

\section{Counterexample}\label{sec:example}

\begin{theorem}
There is a family $\cal F$ of digraphs such that $\bar{s}_{\rm
co}(D)=O(\sqrt{n})$ and $\bar{s}_{\rm cc}(D)=O(\sqrt{n})$ for each
$D\in \cal F$.
\end{theorem}
\begin{proof} For $t=1,2,\ldots$ and $r=\lceil
\sqrt{t}\rceil$, the acyclic digraph $D_t$ consists of vertex set
$V(D_t)=X\cup Y\cup \{z\}\cup Y'\cup X'$, where $$X=\{x_i:\ i\in
[t]\}, X'=\{x'_i:\ i\in [t]\}, Y=\{y_j:\ i\in [r]\}, Y'=\{y'_j:\
i\in [r]\},$$ and arc set
$$ A(D_t)= \{x_ix_{i+1},x'_ix'_{i+1}:\ i\in [t-1]\}\cup
 \{x_ty_j,y_jz,zy'_j,y'_jx'_1:j\in [r] \}.$$

For illustration, see Figure~\ref{fig:few}.
\begin{figure}
\begin{center}
\begin{pspicture}(13.0,4.0)
\psset{linecolor=black} \psset{linewidth=1.5pt}
\psset{unit=9mm}%
\cnode*(0,2.5){0.05}{n4} \cnode*(0.5,2.5){0.05}{n4a}
\cnode*(1,2.5){0.05}{n4b}%
\cnode*(1.5,2.5){0.15}{n5}%
\rput(1.5,2){$x_{t-3}$} \cnode*(2.5,2.5){0.15}{n6}
\rput(2.5,2){$x_{t-2}$}
\cnode*(3.5,2.5){0.15}{n7}%
\rput(3.5,2){$x_{t-1}$}
\cnode*(4.5,2.5){0.15}{n8}%
\rput(4.5,2){$x_t$}

\ncline{->}{n5}{n6}  \ncline{->}{n6}{n7} \ncline{->}{n7}{n8}

\cnode*(5.5,4.5){0.15}{m1}%
\cnode*(5.5,3.5){0.15}{m2}%
\cnode*(5.5,2.){0.05}{m3}%
\cnode*(5.5,2.5){0.05}{m3a} \cnode*(5.5,3.0){0.05}{m3b}
\cnode*(5.5,1.5){0.15}{m4}
\cnode*(5.5,0.5){0.15}{m5}%

\cnode*(6.5,2.5){0.15}{d} \rput(6.5,2){$z$}

 \ncline{->}{n8}{m1} \ncline{->}{n8}{m2} \ncline{->}{n8}{m4}  \ncline{->}{n8}{m5}

\ncline{->}{m1}{d} \ncline{->}{m2}{d} \ncline{->}{m4}{d}
\ncline{->}{m5}{d}

\cnode*(7.5,2.5){0.05}{m6}%
\cnode*(7.5,2.0){0.05}{m6a} \cnode*(7.5,3.0){0.05}{m6b}
\cnode*(7.5,3.5){0.15}{m7}%
\cnode*(7.5,4.5){0.15}{m8}%
\cnode*(7.5,1.5){0.15}{m9}
\cnode*(7.5,0.5){0.15}{m10}%

 \ncline{->}{d}{m7} \ncline{->}{d}{m8}  \ncline{->}{d}{m9}
\ncline{->}{d}{m10}

\cnode*(8.5,2.5){0.15}{z1}%
\rput(8.5,2){$x'_1$}
\cnode*(9.5,2.5){0.15}{z2}%
\rput(9.5,2){$x'_2$}
\cnode*(10.5,2.5){0.15}{z3}%
\rput(10.5,2){$x'_3$}
\cnode*(11.5,2.5){0.15}{z4}%
\rput(11.5,2){$x'_4$}
\cnode*(12,2.5){0.05}{z5}%
\cnode*(12.5,2.5){0.05}{z6}
\cnode*(13,2.5){0.05}{z7}%

 \ncline{->}{m7}{z1} \ncline{->}{m8}{z1}  \ncline{->}{m9}{z1}
\ncline{->}{m10}{z1}

\ncline{->}{z1}{z2} \ncline{->}{z2}{z3} \ncline{->}{z3}{z4}

\end{pspicture}

\end{center}
\caption{\label{fig:few} {Digraphs from $\cal F$}}
\end{figure}
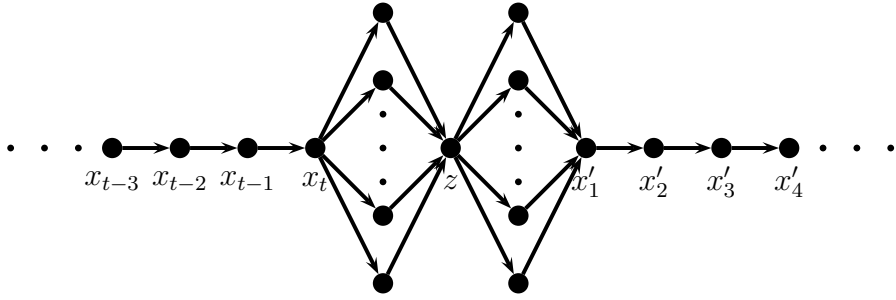

Consider the family $\cal C$ of convex sets of $D_t$ of size at
least $2r+2$. Observe that each set in $\cal C$ contains a vertex in
$X\cup X'.$ Thus, $|{\cal C}|=|{\cal C}_X|+|{\cal C}_{X'}|+|{\cal
C}_{X,X'}|$, where ${\cal C}_X$ (${\cal C}_{X'}$, ${\cal C}_{X,X'}$,
respectively) is the family of sets in $\cal C$ containing a vertex
in $X$ but not in $X'$ (a vertex in $X'$ but not in $X$, vertices in
both $X$ and $X'$, respectively). By symmetry, $|{\cal C}|=2\cdot
|{\cal C}_{X}|+|{\cal C}_{X,X'}|.$ Observe that ${\cal C}_{X}$
consists of

(a) $\Theta(t^2)$ sets in $\cal C$ containing only vertices in $X$;

(b) $\Theta(t2^r)$ sets in $\cal C$ containing only vertices in
$X\cup Y$ but at least one vertex in $Y$;

(c) $\Theta(t2^r)$ sets in $\cal C$ containing $z$ and possibly some
vertices in $Y'$.

Thus, $|{\cal C}_{X}|=\Theta(t2^r).$ In addition, $|{\cal
C}_{X,X'}|=\Theta(t^2)$ and, hence, $|{\cal C}|=\Theta(t2^r).$

For each $Q\subseteq Y$ and $Q'\subseteq Y'$, the set $Q\cup
\{z\}\cup Q'$ is connected and convex. So, there are $2^{2r}$
connected convex sets contained in the set $Y\cup \{z\}\cup Y'.$,
and thus there are $\Omega(2^{2r})$ connected convex sets in
$D_t.$ Hence,
\begin{align*} \bar{s}_{\rm co}(D)&=\frac{\sum_{C\in \sco(D)} |C|}{\co(D)}
\leq \frac{1}{\cc(D)}\left(\sum_{C \in \sco(D) \atop |C| > 2r+1} |C|
\right)+ (2r
+2) \\
& \le O\left(t\frac{t2^r}{2^{2r}}\right)+(2r+2)=O(r),
\end{align*}
and similarly
\[\bar{s}_{\rm cc}(D)\le
O\left(t\frac{t2^r}{2^{2r}}\right)+(2r+2)=O(r).\]
\end{proof}

Let us remark that the digraph in the family $\cal F$ in the
previous proof have asymptotically the same number of convex sets
and connected convex sets. This is generally not true as can be
seen by considering the family of digraphs ${\cal
G}=\{G_1,G_2,\ldots\}$, where the digraph $G_i$ consists of a
source $s$, a sink $t$ and $i$ internally-disjoint directed paths
with two internal vertices each; see Figure~\ref{fig:asym}.
\begin{figure}
\begin{center}
\begin{pspicture}(6.0,4.0)
\psset{linecolor=black} \psset{linewidth=1.5pt}
\psset{unit=9mm}%
\cnode*(0,2){0.15}{s} \cnode*(2,0){0.15}{a1} \cnode*(2,1){0.15}{a2}
\cnode*(2,3){0.15}{a3}%
\cnode*(2,4){0.15}{a4}%
\cnode*(4,0){0.15}{b1} \cnode*(4,1){0.15}{b2}
\cnode*(4,3){0.15}{b3}%
\cnode*(4,4){0.15}{b4}%
\cnode*(6,2){0.14}{t}

\cnode*(3,1.5){0.05}{h1}%
\cnode*(3,2){0.05}{h2}%
\cnode*(3,2.5){0.05}{h3}%

\rput(0,1.5){$s$} \rput(6,1.5){$t$}

\ncline{->}{s}{a1}
\ncline{->}{s}{a2}\ncline{->}{s}{a3}\ncline{->}{s}{a4}
\ncline{->}{a1}{b1}\ncline{->}{a2}{b2}\ncline{->}{a3}{b3}\ncline{->}{a4}{b4}
\ncline{->}{b1}{t}\ncline{->}{b2}{t}\ncline{->}{b3}{t}\ncline{->}{b4}{t}

\end{pspicture}

\end{center}
\caption{\label{fig:asym} {Digraphs from $\cal G$}}
\end{figure}
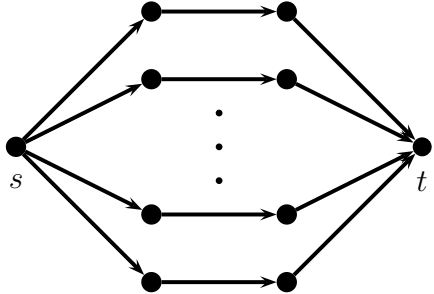
Then any non-empty subset of $V(G_i)\setminus \{s,t\}$ is convex
but only $3i$ of these $4^i-1$ sets are connected. There are $3^i$
connected convex sets containing $s$ but not $t$, $3^i$ connected
convex sets containing $t$ but not $s$ and only one connected
convex set containing $s$ and $t$ (namely $V(G_i)$). Hence $G_i$
contains at least $4^i-1$ convex sets but at most $2\cdot
3^i+3i+1$ connected convex sets and thus $\cc(G_i)/\co(G_i)
\rightarrow 0$ as $i\rightarrow \infty$.

\section{Connected convex sets of size $k$} \label{sec:size}
In this section we show that  any acyclic digraph $D$ of order $n$
contains at least $n-i+1$ connected convex sets of order $i$ for
each $i\in [n]$. To do so we first show that every connected
convex set that is not the entire graph  can be extended by one
vertex, and that there always exists a source or a sink that is
not a cut-vertex.

\begin{lemma}\label{lem:extension}
Let $D=(V,A)$ be a connected acyclic digraph and let $H\neq V$ be a
connected convex set in $D$. Then there exists a vertex $v$ in
$V\setminus H$ such that $H\cup \{v\}$ is a connected convex set in
$D$.
\end{lemma}
\begin{proof}
Since $D$ is connected there is an arc $uv$ with either $u\in
V\setminus H$ and $v \in  H$, or $u \in  H$ and $v\in V\setminus H$.
We may assume that $u \in  H$ and $v\in V\setminus H$ as the other
case can be treated similarly. Consider a longest directed path $P$
in which the initial vertex is in $H$, the terminal vertex is $v$
and all other vertices are in $V\setminus H$. Let $y$ and $w$ be the
initial and second vertices in $P$, respectively. We will prove that
$H\cup \{w\}$ is a connected convex set. Note that $H\cup \{w\}$ is
connected as $H$ is connected and there is an arc $yw$. Let $x$ be
an arbitrary vertex of $H$. Observe that it suffices to show that
there is no directed path between $x$ and $w$ having intermediate
vertices that are all not in $H\cup \{w\}$.

Suppose that there is a directed path $Q$ from $w$ to $x$. Since $D$
is acyclic, $y\neq x$ and we obtain a directed path $yQ$ between $y$
and $x$ containing  $w$ and thus vertices outside $H$, a
contradiction with the convexity of $H$. Suppose that there is a
directed path $R$ from $x$ to $w$ having intermediate vertices that
are not in $H\cup \{w\}$. If $R$ and $P$ have a vertex $z$ in common
then there is a closed walk from $z$ to itself via $w$, which is
impossible as $D$ is acyclic. If $R$ and $P$ have no vertex in
common, then $R(P-y)$ is a directed path from $x$ to $v$ with all
intermediate vertices in $V\setminus H$, and  it is longer than $P$.
This contradicts the choice of $P$.
\end{proof}

\begin{lemma}\label{lem:cut}
Let $D$ be a connected  acyclic digraph of order $n\geq 2$. Then
there exist at least two non-cut-vertices that are a source or a
sink.
\end{lemma}
\begin{proof}
We prove the result by induction on $n$. For $n=2$ the result
clearly holds as both vertices are non-cut-vertices and one is a
source and one is a sink. Now assume that the lemma is true for
all connected acyclic digraphs of order at most $n_0\geq 2$, and
let $D$ be a connected acyclic digraph of order $n_0+1$. Since $D$
is acyclic there exists a source $s$  and sink $t$ of $D$; see for
example \cite{bang2000}. If $s$ and $t$ are non-cut-vertices, then
the claim of the lemma is satisfied. So assume that one of $s$ and
$t$, say $s$, is a cut-vertex. Let $C_1,\ldots,C_k$ be the vertex
sets of the components of $D-s$. Note that $k\geq 2$ and that for
all $i=1,\ldots,k$, $D[C_1+s]$ is a connected acyclic digraph of
order at least $2$ and at most $n_0$. Thus, by the induction
hypothesis each of the digraphs $D[C_1+s]$ and $D[C_2+s]$ has at
least two non-cut-vertices that are either a source or a sink.
Choose two of these vertices that are not equal to $s$. It is
easily verified that these vertices are non-cut-vertices and are
sources or sinks of $D$.
\end{proof}

\begin{theorem}
Let $D$ be a connected acyclic digraph of order $n$. Then there
exist at least $n-i+1$ connected convex sets in $D$ of order $i$ for
each $i\in [n]$.
\end{theorem}
\begin{proof}
We prove the result by induction on $n$. For $n=1$ the result is
trivial. Assume that the claim of the theorem holds for all digraphs
of order $1\leq n\leq n_0$. We will show the result for an arbitrary
connected acyclic digraphs $D$ of order $n_0+1$.  Since $D$ is
connected and acyclic, by Lemma~\ref{lem:cut} there exists a source
or a sink $s$ of $D$ that is not a cut-vertex. By the induction
hypothesis $D- s$ contains $n_0-i+1$ connected convex sets of size
$i=1,2,\ldots,n_0$. These sets are also connected convex sets of $D$
as $s$ is a source or sink and thus no directed path between
vertices in $D- s$ can use $s$. In addition by a using
Lemma~\ref{lem:extension} $i-1$ times we can deduce that there is at
least one connected convex set in $D$ of size $i$ containing $s$.
Thus, there exist at least $n_0-i+1+1= (n_0+1)-i+1$ sets of size $i$
for each $i\in[n_0]$. Since the vertex set of the digraph is a
connected convex set there is also (exactly) one such set of order
$n_0+1$.
\end{proof}

Let us remark that the directed path of order $n$ has exactly
$n-i+1$ connected convex subgraphs of order $i$.

\section{Open Question}

Recall that the authors of \cite{gjrssy07} designed an $O(n\cdot
cc(D))$-time algorithm to generate all connected convex sets of a
connected acyclic digraph $D$. As we showed in Section
\ref{sec:example} there could be an asymptotically faster algorithm.
Is there an $O(\sum_{C\in \scc(D)}|C|)$-time algorithm? It is known
\cite{ps} that there is an algorithm for generating all (not
necessarily connected) convex sets of a connected acyclic digraph
$D$ in $O(\sum_{C\in \sco(D)}|C|)$ time.

\vspace{3mm}

\noindent{\bf Acknowledgement} Research of Gutin was supported in
part by an EPSRC grant.

\end{document}